\documentclass[twocolumn]{IEEEtran}
\usepackage{amsmath}
\usepackage{amssymb}
\usepackage{mathrsfs}
\usepackage{cite}
\usepackage{epsfig}
\usepackage{epsfig}
\usepackage{graphics}
\usepackage{hyperref}
\usepackage{epsfig}

\newtheorem{thm}{Theorem}
\newtheorem{lem}{Lemma}

\newtheorem{defi}{Definition}

\newtheorem{rem}{Remark}

\begin{document}

\title{On the Secure DoF of the Single-Antenna  MAC}
\author{Ghadamali Bagherikaram, Abolfazl S. Motahari, Amir K. Khandani\\
Coding and Signal Transmission Laboratory,
University of Waterloo, Ontario, Canada\\
 Emails: \{gbagheri,abolfazl,khandani\}@cst.uwaterloo.ca
}
 \maketitle
\footnotetext{Financial support provided by Nortel, and the
corresponding matching funds by the Natural Sciences and Engineering
Research Council of Canada (NSERC), and Ontario Ministry of Research
\& Innovation (ORF-RE) are gratefully acknowledged.}
\begin{abstract}
A new achievability rate region for the secure discrete memoryless Multiple-Access-Channel (MAC) is presented. Thereafter, a novel secure coding scheme is proposed to achieve a positive Secure Degrees-of-Freedom (S-DoF) in the single-antenna MAC. This scheme converts the single-antenna system into a multiple-dimension system with fractional dimensions. The achievability scheme is based on the alignment of signals into a small sub-space at the eavesdropper, and the simultaneous separation of the signals at the intended receiver. Tools from
the field of Diophantine Approximation in number theory are used to analyze the probability of error in the coding scheme.
\end{abstract}
\section{Introduction}
The notion of information theoretic secrecy in communication systems
was first introduced in \cite{1}. The information
theoretic secrecy requires that the received signal of the
eavesdropper does not provide any information about the transmitted
messages. Following the pioneering works
of \cite{2} and \cite{3} which studied the
wiretap channel, many multi-user channel models have been
considered from a perfect secrecy point of view.
The secure Gaussian multiple-access-channel with an external eavesdropper
is introduced in \cite{4,5}. This channel consists of an ordinary Gaussian MAC and an
external eavesdropper. The capacity region of this channel is still an open problem
in the information theory field. For this channel, an achievable rate scheme based on Gaussian codebooks
is proposed in \cite{4}, where also the
sum secrecy capacity of the degraded Gaussian channel is
also found.

On the other hand, it is shown that the random coding argument may be insufficient to prove
capacity theorems for certain channels; instead, structure codes can be used to construct efficient channel
codes for Gaussian channels. In reference \cite{6}, nested lattice codes are used to provide secrecy in two-user Gaussian channels. In \cite{6} it is shown that structure codes can achieve a positive S-DoF in a two-user MAC. Especially, the achievability scheme of \cite{6} provides a S-DoF of $\frac{1}{2}$ for a small category of channel gains and for the other categories, it provides a S-DoF of strictly less than $\frac{1}{2}$.

In reference \cite{7}, the concept of interference alignment is introduced and has illustrated its capability
in achieving the full DoF of a class of two-user X channels.
In reference \cite{8}, a novel coding scheme applicable in networks with single antenna nodes is proposed. This scheme converts a single antenna system into an equivalent Multiple Input Multiple Output (MIMO) system with fractional dimensions.
In this work we use the notion of \emph{real alignment} of \cite{8} to prove that for almost all channel gains in the secure $K$ user single-antenna Gaussian MAC, we can achieve the S-DoF of $\frac{K-1}{K}$. The scheme of this work differs from that of \cite{6}, in the sense that our scheme achieves the S-DoF of $\frac{1}{2}$ for almost all channel gains.Therefore, the carve of S-DoF versus channel gains is almost certainly constant.

The rest of the paper is organized as follows: Section II provides some background and preliminaries. In section III, we present our results for the achievable S-DoF of the single-antenna MAC. Finally, section IV concludes the paper.
\section{Preliminaries}
Consider a secure $K$-user Gaussian single-antenna Multiple-Access-Channel (MAC). 
In this confidential setting, each user $k$ ($k\in \mathcal{K}\stackrel{\triangle}{=}\{1,2,...,K\}$) wishes to send the message $W_{k}$ to the intended receiver in $n$ uses of the channel, simultaneously and prevent the eavesdropper from having any information about the messages. At a specific time, the signals received by the intended receiver and the eavesdropper is given by
\begin{IEEEeqnarray}{rl}
Y&=\sum_{k=1}^{K}h_{k}X_{k}+\widetilde{W}_{1}\\ \nonumber
Z&=\sum_{k=1}^{K}h_{k,e}X_{k}+\widetilde{W}_{2},
\end{IEEEeqnarray}
where $X_{k}$ for $k\in\mathcal{K}$ is a real input scalar under an input average power constraint. We require that $E[X_{k}^{2}]\leq P$, $Y$ and $Z$ be real output scalars which are received by the destination and the eavesdropper, respectively, $h_{k}$ and $h_{k,e}$ for $k=1,2,...,K$ are fixed, real scalars which model the channel gains between the transmitters and the intended receiver and the eavesdropper, respectively. The channel state information is assumed to be known perfectly at all the transmitters and at all receivers, and $\widetilde{W}_{1}$, $\widetilde{W}_{2}$  are real Gaussian random variables with zero means and unit variances.
Let $X_{k}^{n}$, $Y^{n}$ and $Z^{n}$ denote the random channel inputs and random channel outputs over a block of $n$ samples. Furthermore, let $\widetilde{W}^{n}_{1}$,
and $\widetilde{W}^{n}_{2}$ denote the additive noises of the channels.
The elements of
$\widetilde{W}^{n}_{1}$ and $\widetilde{W}^{n}_{2}$ are independent zero mean Gaussian random variables with unit variances. In addition, $\widetilde{W}^{n}_{1}$ and $\widetilde{W}^{n}_{2}$ are independent of $X_{k}^{n}$'s and $W_{k}$'s. A
$((2^{nR_{1}},2^{nR_{2}},...,2^{nR_{k}}),n)$ secret code for the above channel consists of
the following components:

\emph{1}) $K$ secret message sets $\mathcal{W}_{k}=\{1,2,...,2^{nR_{k}}\}$.

\emph{2}) $K$  stochastic encoding functions $f_{k}(.)$ which map the secret messages to the transmitted symbols, i.e., $f_{k}: w_{k}\rightarrow X^{n}_{k}$ for each $w_{k}\in\mathcal{W}_{k}$. At encoder $k$, each codeword is designed according to the transmitter's average power constraint $P$.

\emph{3}) A decoding function $\phi(.)$ which maps the received symbols to estimate the messages: $\phi(Y^{n})\rightarrow (\hat{W_{1}},...,\hat{W_{K}}) $.

The reliability of the transmission is measured by the average probability of error, which is defined as the probability that the decoded messages are not equal
to the transmitted messages; that is
\begin{IEEEeqnarray}{rl}
P_{e}^{(n)}&=\frac{1}{\prod_{k=1}^{K}2^{nR_{k}}}\sum_{(w_{1},...,w_{K})\in\mathcal{W}_{1}\times....\times\mathcal{W_{K}}}\\ \nonumber &P\left\{\phi\left(Y^{n}\right)\neq \left(w_{1},...,w_{K}\right)|\left(w_{1},...,w_{k}\right)~ \hbox{is sent}\right\}.
\end{IEEEeqnarray}
The secrecy level is measured by a normalized equivocation, defined as follows: The normalized equivocation for each subset of messages $W_{\mathcal{S}}$ for $\mathcal{S}\subseteq\mathcal{K}$ is
\begin{equation}
\Delta_{\mathcal{S}}\stackrel{\triangle}{=}\frac{H(W_{\mathcal{S}}|Z^{n})}{H(W_{\mathcal{S}})}.
\end{equation}
The rate-equivocation tuple $(R_{1},...,R_{K},d)$ is said to be achievable for the Gaussian single-antenna multiple-access-channel with confidential messages, if for
any $\epsilon>0$, there
exists a sequence of $((2^{nR_{1}},...,2^{nR_{K}}),n)$ secret codes, such that
for sufficiently large $n$,
\begin{equation}
P_{e}^{(n)}\leq\epsilon,
\end{equation}
and
\begin{equation}
\Delta_{\mathcal{S}}\geq d-\epsilon,~~~~~~~\forall\mathcal{S}\subseteq\mathcal{K}.
\end{equation}
The perfect secrecy rate tuple $(R_{1},...,R_{K})$ is said to be achievable when $d=1$. When all the transmitted messages are perfectly secure, we have
\begin{equation}\label{eq1}
\Delta_{\mathcal{K}}\geq 1-\epsilon,
\end{equation}
In \cite{5} it is shown that when all of the $K$ messages are perfectly secure, then it guarantees that any subset of the messages become perfectly secure.

The total Secure Degrees-of Freedom (S-DoF) of $\eta$ is said to be achievable if the rate-equivocation tuple $(R_{1},...,R_{K},d=1)$ is achievable, and
\begin{equation}
\eta=\lim_{P\rightarrow\infty}\frac{\sum_{k=1}^{K}R_{k}}{\frac{1}{2}\log P}
\end{equation}
\section{S-DoF of the Single-Antenna MAC}
In this section, we consider the achievable S-DoF of the Gaussian single-antenna multiple-access-channel under the perfect secrecy constraint. In order to satisfy the perfect secrecy constraint, we use the random binning coding scheme to generate the codebooks. To maximize the achievable degrees of
freedom, we adopt the signal alignment scheme used in \cite{8} to separate the signals at the intended receiver and simultaneously align the signals into a small subspace at the eavesdropper.
The main results of this section are presented in the following theorems.
First, we present an achievable secrecy rate region for the discrete memoryless multiple-access-channel.
\begin{thm}\label{th1}
For the perfectly secure discrete memoryless multiple-access-channel of $P(y,z|x_{1},...,x_{K})$, the region of
\begin{IEEEeqnarray}{rl}\label{eq2}
\Big\{(R_{1},...,R_{K})|&\sum_{i\in\mathcal{S}}R_{i}\leq I(U_{\mathcal{S}};Y|U_{\mathcal{S}^{c}}),~~~\forall \mathcal{S}\subset\mathcal{K},\\ \nonumber&\sum_{k\in\mathcal{K}}R_{k}\leq\left[I(U_{\mathcal{K}};Y)-I(U_{\mathcal{K}};Z)\right]^{+}\Big\},
\end{IEEEeqnarray}
for any distribution of
\begin{IEEEeqnarray}{rl}\nonumber
&P(u_{1})P(u_{2})....P(u_{K})P(x_{1}|u_{1})P(x_{2}|u_{2})...P(x_{K}|u_{K})\\ \nonumber
&\times P(y,z|x_{1},...,x_{K}),
\end{IEEEeqnarray}
 is achievable.
\end{thm}
\begin{proof}
The proof is available in the Appendix.
\end{proof}
Reference \cite{5} derived an achievable rate region with Gaussian codebooks and power
control for the degraded Gaussian secure multiple-access-channel, when all the transmitters
and receivers are equipped with a single antenna. Theorem \ref{th1}, however, gives an
achievability secrecy rate region for the general discrete memoryless multiple-access-channel.
Our achievability rate region is also larger than the region of \cite{5} in a special case of the Gaussian degraded case.
When the transmitters and receivers are equipped with a single antenna, the total achieved S-DoF by using Gaussian codebooks is $0$.
Here, we will provide a coding scheme based on integer codebooks, and show that for almost all channel gains a positive total S-DoF is achievable. The following theorem illustrates our results.
\begin{thm}\label{th3}
For the Gaussian single antenna multiple-access-channel, a total $\frac{K-1}{K}$ secure degrees-of-freedom can be achieved for almost all channel gains.
\end{thm}
\begin{proof}
Let us define $\widetilde{X}_{k}\stackrel{\triangle}{=}\frac{h_{k,e}}{A}X_{k}$ and $\widetilde{h}_{k}\stackrel{\triangle}{=}\frac{h_{k}}{h_{k,e}}$.
Without loss of generality let us assume that $\widetilde{h}_{K}=1$; the cannel model then is equivalent as follows:
\begin{IEEEeqnarray}{rl}
Y&=A\left[\sum_{k=1}^{K-1}\widetilde{h}_{k}\widetilde{X}_{k}+\widetilde{X}_{K}\right]+\widetilde{W}_{1}\\ \nonumber
Z&=A\sum_{k=1}^{K}\widetilde{X}_{k}+\widetilde{W}_{2},
\end{IEEEeqnarray}
where, $A^{2}E[\widetilde{X}_{k}^{2}]\leq\widetilde{P}\stackrel{\triangle}{=} h_{k,e}^{2} P$. In this model we say that the signals are aligned at the eavesdropper, according to the following definition:
\begin{defi}
The signals $\widetilde{X}_{1}$, $\widetilde{X}_{2}$,...,$\widetilde{X}_{K}$ are said to be aligned at a receiver
if its received signal is a rational combination of them.
\end{defi}
Note that, in $n$-dimensional Euclidean spaces ($n\geq 2$), two signals are aligned when they are received in the
same direction at the receiver. In general, $m$ signals are aligned at a receiver if they span a subspace with
a dimension less than $m$. The above definition, however, generalizes the concept of alignment for the one-dimensional real numbers. Our coding scheme is based on integer codebooks, which means that $\widetilde{X}_{k}\in \mathbb{Z}$ for all $k\in\mathcal{K}$. If some integer signals are aligned at a receiver, then their effect is similar to a single signal at high SNR regimes. This is due to the fact that rational numbers form a filed; therefore the sum of constellations from $\mathbb{Q}$ form a constellation in $\mathbb{Q}$ with an enlarged cardinality.

Before we present our achievability scheme, we need to define the rational dimension of a set of real numbers.
\begin{defi}(\emph{Rational Dimension})
The rational dimension of a set of real numbers $\{\widetilde{h}_{1},\widetilde{h}_{2},...,\widetilde{h}_{K-1},\widetilde{h}_{K}=1\}$ is
$M$ if there exists a set of real numbers $\{g_{1},g_{2},...,g_{M}\}$, such that each $\widetilde{h}_{k}$ can be represented as a rational combination of $g_{i}$'s, i.e., $\widetilde{h}_{k}=a_{k,1}g_{1}+a_{k,2}g_{2}+...+a_{k,M}g_{M}$, where $a_{k,i}\in\mathbb{Q}$ for all $k\in\mathcal{K}$ and $i\in\mathcal{M}$.
\end{defi}
In fact, the rational dimension of a set of channel gains is the effective dimension seen at the corresponding receiver. In particular, $\{\widetilde{h}_{1},\widetilde{h}_{2},...,\widetilde{h}_{K}\}$ are \emph{rationally independent} if the rational dimension is $K$, i.e., none of the $\widetilde{h}_{k}$ can be represented as the rational combination of other numbers.

Note that all of the channel gains $\widetilde{h}_{k}$ are generated independently with a distribution.
From the number theory, it is known that the set of all possible channel gains that are rationally independent has a Lebesgue measure $1$. Therefore, we can assume that $\{\widetilde{h}_{1},\widetilde{h}_{2},...,\widetilde{h}_{K}\}$ are rationally independent, almost surely. Our achievability coding scheme is as follows:
\subsubsection{Encoding}
Each transmitter limits its input symbols to a finite set, which is called the transmit constellation. Even
though it has access to the continuum of real numbers, restriction to a finite set has the benefit of easy
and feasible decoding at the intended receiver. The transmitter $k$ selects a constellation $\mathcal{V}_{k}$ to
send message $W_{k}$. The constellation points are chosen from integer points, i.e., $\mathcal{V}_{k}\subset\mathbb{Z}$. We assume that $\mathcal{V}_{k}$ is a bounded set. Hence, there is a constant $Q_{k}$ such that $\mathcal{V}_{k}\subset [-Q_{k},Q_{k}]$. The cardinality of $\mathcal{V}_{k}$ which limits the rate of message $W_{k}$ is denoted by $\|\mathcal{V}_{k}\|$.

Having formed the constellation, the transmitter $k$ constructs a random codebook for message $W_{k}$  with
rate $R_{k}$. This can be accomplished by choosing a probability distribution on the input alphabets. The
uniform distribution is the first candidate and it is selected for the sake of simplicity. Therefore, the stochastic
encoder $k$ generates $2^{n(I(\widetilde{X}_{k};Y|\widetilde{X}_{(\mathcal{K}-k)^{c}})+\epsilon_{k})}$ independent and identically
distributed sequences $\widetilde{x}_{k}^{n}$ according to the
distribution $P(\widetilde{x}_{k}^{n})=\prod_{i=1}^{n}P(\widetilde{x}_{k,i})$, where $P(\widetilde{x}_{k,i})$ denotes
the probability distribution function of the uniformly distributed random variable $\widetilde{x}_{k,i}$ over
$\mathcal{V}_{k}$. Next, these sequences are randomly distributed into $2^{nR_{k}}$ bins. Index each of the
bins by $w_{k}\in\{1,2,...,2^{nR_{k}}\}$.

For each user $k\in\mathcal{K}$, to send message $w_{k}$, the
transmitter looks for a $\widetilde{x}_{k}^{n}$ in bin $w_{k}$. The rates are such that there exist more than one
$\widetilde{x}_{k}^{n}$. The transmitter randomly chooses one of them and sends $x_{k}^{n}=A\frac{\widetilde{x}_{k}^{n}}{h_{k,e}}$.
The parameter $A$ controls the input power.
\subsubsection{Decoding}
At a specific time, the received signal at the legitimate receiver is as follows:
\begin{IEEEeqnarray}{rl}\nonumber
Y=A\left[\widetilde{h}_{1}\widetilde{X}_{1}+\widetilde{h}_{2}\widetilde{X}_{2}+...+\widetilde{h}_{K-1}\widetilde{X}_{K-1}+\widetilde{X}_{K}\right]+\widetilde{W}_{1}
\end{IEEEeqnarray}
The legitimate receiver passes the received signal $Y$ through a hard decoder. The hard decoder looks for
a point $\widetilde{Y}$ in the received constellation
$\mathcal{V}_{r}=A\left[\widetilde{h}_{1}\mathcal{V}_{1}+\widetilde{h}_{2}\mathcal{V}_{2}+...+\widetilde{h}_{K-1}\mathcal{V}_{K-1}+\mathcal{V}_{K}\right]$
which is the nearest point to  the received signal $Y$. Therefore, the continuous channel changes to a discrete one in
which the input symbols are taken from the transmit constellations $\mathcal{V}_{k}$ and the output symbols belonging
to the received constellation $\mathcal{V}_{r}$. $\widetilde{h}_{k}$'s are rationally independent. This means that
the equation $A\left[\widetilde{h}_{1}X_{1}+\widetilde{h}_{2}X_{2}+...+\widetilde{h}_{K-1}X_{K-1}+X_{K}\right]=0$ has no rational solution.
This property implies that any real number $v_{r}$ belonging to the constellation $\mathcal{V}_{r}$ is uniquely decomposable as $v_{r}=A\sum_{k=1}^{K}
\widetilde{h}_{k}\widehat{\widetilde{X}}_{k}$. Note that if there exists another possible decomposition
$\widetilde{v}_{r}=A\sum_{k=1}^{K}
\widetilde{h}_{k}\widehat{\widetilde{X}}_{k}^{'}$, then $\widetilde{h}_{k}$'s have to be rationally-dependent, which
is a contradiction. We refer to this property as property $\Gamma$. This property in fact implies that if there is no additive noise
in the channel, then the receiver can decode all the transmitted signals with zero error probability.
\begin{rem}
In a random environment, it is easy to show that the set of channel gains which are rationally-dependent has a measure of zero, with respect to the Lebesgue measure. Therefore, Property $\Gamma$ is satisfied almost surely.
\end{rem}
\subsubsection{Error Probability Analysis}
Let $d_{\min}$ denote the minimum distance in the received constellation $\mathcal{V}_{r}$. Having property $\Gamma$, the
receiver can decode the transmitted signals. Let $V_{r}$ and $\hat{V}_{r}$ be the transmitted and decoded  symbols, respectively. The
probability of error i.e., $P_{e}=P(\hat{V}_{r}\neq V_{r})$, is bounded as follows:
\begin{equation}\label{eq10}
P_{e}\leq Q(\frac{d_{\min}}{2})\leq \exp(-\frac{d_{\min}^{2}}{8})
\end{equation}
where $Q(x)=\frac{1}{\sqrt{2\pi}}\int_{x}^{\infty}\exp(-\frac{t^{2}}{2})dt$. Note that finding $d_{\min}$ is generally not easy. Using Khintchine and Groshev theorems, however, it is possible to lower bound the minimum distance. Here we explain
some backgrounds to use the theorems of Khintchine and Groshev.

The field of Diophantine approximation in number theory deals with approximation of real numbers
with rational numbers. The reader is referred to \cite{9,10} and the references therein. The Khintchine
theorem is one of the cornerstones in this field. This theorem provides a criteria for a given function $\psi:\mathbb{N}\rightarrow\mathbb{R}_{+}$
and real number $h$, such that $|p + \widetilde{h}q| < \psi(|q|)$ has either infinitely many solutions or at most finitely
many solutions for $(p,q)\in\mathbb{Z}^{2}$. Let $\mathcal{A}(\psi)$ denote the set of real numbers, such that $|p +\widetilde{h}q| < \psi(|q|)$ has
an infinite number of solutions in integers. The theorem has two parts. The first part is the convergent part and
states that if $\psi(|q|)$ is convergent, i.e.,
\begin{equation}
\sum_{q=1}^{\infty}\psi(q)< \infty
\end{equation}
then $\mathcal{A}(\psi)$ has a measure of zero with respect to the Lebesgue measure. This part can be rephrased in more
convenient way, as follows: For almost all real numbers, $|p+\widetilde{h}q| > \psi(|q|)$ holds for all $(p, q) \in\mathbb{Z}^{2}$ except
for a finite number of them. Since the number of integers violating the inequality is finite, one can find a
constant $c$ such that
\begin{equation}
|p + \widetilde{h}q| > c\psi(|q|)
\end{equation}
holds for all integers $p$ and $q$, almost surely. The divergent part of the theorem states that $\mathcal{A}(\psi)$ has the
full measure, i.e. the set $\mathbb{R} - \mathcal{A}(\psi)$ has measure of zero provided that $\psi$ is decreasing and $\psi(|q|)$ is divergent,
i.e.,
\begin{equation}
\sum_{q=1}^{\infty}\psi(q)=\infty.
\end{equation}
There is an extension to Khintchine's theorem which regards the approximation of linear forms. Let
$\mathbf{\widetilde{h}} = (\widetilde{h}_{1}, \widetilde{h}_{2},...,\widetilde{h}_{K-1})$ and $\mathbf{q} = (q_{1}, q_{2},... , q_{K-1})$ denote $(K-1)$-tuples
in $\mathbb{R}^{K-1}$ and $\mathbb{Z}^{K-1}$, respectively.
Let $\mathcal{A}_{K-1}(\psi)$ denote the set of $(K-1)$-tuple real numbers $\mathbf{\widetilde{h}}$ such that
\begin{equation}
|p + q_{1}\widetilde{h}_{1} + q_{2}\widetilde{h}_{2} +... +q_{K-1}\widetilde{h}_{K-1}| <\psi(|\mathbf{q}|_{\infty})
\end{equation}
has infinitely many solutions for $p\in \mathbb{Z}$ and $\mathbf{q}\in\mathbb{Z}^{K-1}$. Here, $|\mathbf{q}|_{\infty}$ is the supreme norm of $\mathbf{q}$ which is defined as $\max_{k} |q_{k}|$.
The following theorem illustrates the Lebesgue measure of the set $\mathcal{A}_{K-1}(\psi)$.
\begin{thm}(Khintchine-Groshev)
Let $\psi:\mathbb{N}\rightarrow \mathbb{R}_{+}$. Then, the set $\mathcal{A}_{K-1}(\psi)$ has a measure of zero provided that
\begin{equation}\label{eq9}
\sum_{q=1}^{\infty}q^{K-2}\psi(q)<\infty
\end{equation}
and has the full measure if
\begin{equation}
\sum_{q=1}^{\infty}q^{K-2}\psi(q)=\infty~~~~\hbox{and $\psi$ is monotonic}
\end{equation}
\end{thm}
In this paper, we are interested in the convergent part of the theorem. Moreover, given an arbitrary $\epsilon > 0$ the
function $\psi(q) =\frac{1}{q^{K-1+\epsilon}}$ satisfies the condition of (\ref{eq9}). In fact, the convergent part of the above theorem can be
stated as follows: For almost all $K-1$-tuple real numbers $\mathbf{\widetilde{h}}$ there exists a constant $c$, such that
\begin{equation}\nonumber
|p + q_{1}\widetilde{h}_{1} + q_{2}\widetilde{h}_{2} + . . . + q_{K-1}\widetilde{h}_{K-1}| >\frac{c}{(\max_{k} |q_{k}|)^{K-1+\epsilon}}
\end{equation}
holds for all $p \in\mathbb{Z}$ and $\mathbf{q}\in\mathbb{Z}^{K-1}$.
The Khintchine-Groshev theorem can be used to bound the minimum distance of points in the received
constellation $\mathcal{V}_{r}$. In fact, a point in the received constellation has a linear form of
$v_{r} =A\left[\widetilde{h}_{1}v_{1}+\widetilde{h}_{2}v_{2}+...+\widetilde{h}_{K-1}v_{K-1}+v_{K}\right]$. Therefore, we can conclude that
\begin{equation}
d_{\min} >\frac{Ac}{(\max_{k\in\{1,2,...,K-1\}}Q_{k})^{K-1+\epsilon}}.
\end{equation}
The probability of error in hard decoding, see (\ref{eq10}), can be bounded as:
\begin{equation}\label{eq11}
P_{e} < \exp\left(-\frac{(Ac)^{2}}{8(\max_{k\in\{1,2,...,K-1\}}Q_{k})^{2K-2+2\epsilon}}\right)
\end{equation}
Let us assume that $Q_{k}$ for all $k\in\{1,2, . . . ,K-1\}$ is $Q=\lfloor\widetilde{P}^{\frac{1-\epsilon}{2(K+\epsilon)}}\rfloor$.
Moreover, since $E[\widetilde{X}_{k}^{2}]\leq A^{2}Q_{k}^{2}\leq \widetilde{P}$, we can choose $A=\widetilde{P}^{\frac{K-1+2\epsilon}{2(K+\epsilon)}}$.
Substituting in (\ref{eq11}) yields
\begin{equation}
P_{e} < \exp (-\frac{c^{2}}{8}\widetilde{P}^{\epsilon}).
\end{equation}
Thus, $P_{e}\rightarrow 0$ when $\widetilde{P}\rightarrow \infty$ or equivalently $P\rightarrow \infty$.

\subsubsection{Equivocation Calculation}
Since the equivocation analysis of Theorem \ref{th1} is valid for any input distribution, the integer inputs satisfy the perfect secrecy constraint.

\subsubsection{S-DoF Calculation}
The maximum achievable sum rate is as follows:
\begin{IEEEeqnarray}{rl}
\sum_{k\in\mathcal{K}}R_{k}&= I(\widetilde{X}_{1},\widetilde{X}_{2},...,\widetilde{X}_{K};Y)-I(\widetilde{X}_{1},\widetilde{X}_{2},...,\widetilde{X}_{K};Z)\\ \nonumber
&=H(\widetilde{X}_{1},\widetilde{X}_{2},...,\widetilde{X}_{K}|Z)-H(\widetilde{X}_{1},\widetilde{X}_{2},...,\widetilde{X}_{K}|Y)\\ \nonumber
&\stackrel{(a)}{\geq}H(\widetilde{X}_{1},\widetilde{X}_{2},...,\widetilde{X}_{K}|Z)-1-P_{e}\log\|\widetilde{\mathcal{X}}\|\\ \nonumber
&\stackrel{(b)}{\geq}H(\widetilde{X}_{1},\widetilde{X}_{2},...,\widetilde{X}_{K}|\sum_{k\in\mathcal{K}}\widetilde{X}_{k})-1-P_{e}\log\|\widetilde{\mathcal{X}}\|\\ \nonumber
&\stackrel{(c)}{=}\sum_{k\in\mathcal{K}}H(\widetilde{X}_{k})-H(\sum_{k\in\mathcal{K}}\widetilde{X}_{k})-1-P_{e}\log\|\widetilde{\mathcal{X}}\|\\ \nonumber
&\stackrel{(d)}{=}K\log(2Q+1)-\log(2KQ+1)-1-P_{e}\log\|\widetilde{\mathcal{X}}\|,
\end{IEEEeqnarray}
where $(a)$ follows from Fano's inequality, $(b)$ follows from the fact that conditioning always decreases entropy, $(c)$
follows from chain rule, and $(d)$ follows from the fact that $\widetilde{X}_{k}$ has uniform distribution over $\mathcal{V}_{k}=[-Q,Q]$.
The S-DoF can therefore be computed as follows:
\begin{IEEEeqnarray}{rl}
\eta=\lim_{P\rightarrow \infty}\frac{\sum_{k\in\mathcal{K}}R_{k}}{\frac{1}{2}\log P}=\frac{(K-1)(1-\epsilon)}{K+\epsilon}
\end{IEEEeqnarray}
Since $\epsilon$ can be arbitrary small, then $\eta=\frac{K-1}{K}$ is indeed achievable.
\end{proof}
\section{Conclusion}
Unlike \cite{6}, we presented a coding scheme that can achieve the total amount of $\frac{K-1}{K}$ S-DoF for the $K$ user secure MAC for almost all channel gains. Our scheme is based on a single layer integer coding and random binning.
\appendix
\subsection{Proof of Theorem \ref{th1}}
1) \textit{Codebook Generation}:
The structure of the encoder for user $k\in\mathcal{K}$ is
as follows: Fix $P(u_{k})$ and $P(x_{k}|u_{k})$. The stochastic
encoder $k$ generates $2^{n(I(U_{k};Y|U_{(\mathcal{K}-k)^{c}})+\epsilon_{k})}$ independent and
identically distributed sequences $u_{k}^{n}$ according to the
distribution $P(u_{k}^{n})=\prod_{i=1}^{n}P(u_{k,i})$. Next,
randomly distribute these sequences into $2^{nR_{k}}$ bins. Index each of the
bins by $w_{k}\in\{1,2,...,2^{nR_{k}}\}$.

2) \textit{Encoding}: For each user $k\in\mathcal{K}$, to send message $w_{k}$, the
transmitter looks for a $u_{k}^{n}$ in bin $w_{k}$. The rates are such that there exist more than one
$u_{k}^{n}$. The transmitter randomly chooses one of them and then generates $x_{k}^{n}$ according
to $P(x_{k}^{n}|u_{k}^{n})=\prod_{i=1}^{n}P(x_{k,i}|u_{k,i})$ and sent it.

3) \textit{Decoding}: The received signals at the legitimate
receiver, $y^{n}$, is the output of the
channel $P(y^{n}|x_{\mathcal{K}}^{n})=\prod_{i=1}^{n}P(y_{i}|x_{\mathcal{K},i})$.
The legitimate receiver looks for the unique sequence $u_{\mathcal{K}}^{n}$ such
that $(u_{\mathcal{K}}^{n},y^{n})$ is jointly typical and declares the
indices of the bins containing $u_{k}^{n}$ as the messages received.

4) \textit{Error Probability Analysis}: Since the region of
(\ref{eq2}) is a subset of the capacity region of the multiple-access-channel without secrecy constraint, then the error probability analysis is straightforward and omitted here.

5) \textit{Equivocation Calculation}: To satisfy the perfect secrecy constraint, we need to prove the requirement of (\ref{eq1}). From $H(W_{\mathcal{K}}|Z^{n})$ we have
\begin{IEEEeqnarray}{rl}
H(W_{\mathcal{K}}|Z^{n}) &= H(W_{\mathcal{K}},Z^{n})-H(Z^{n})\\
\nonumber &= H(W_{\mathcal{K}},U_{\mathcal{K}}^{n},Z^{n})- H(U_{\mathcal{K}}^{n}|W_{\mathcal{K}},Z^{n})\\ \nonumber &-H(Z^{n})\\
\nonumber &=H(W_{\mathcal{K}},U_{\mathcal{K}}^{n})+ H(Z^{n}|W_{\mathcal{K}},U_{\mathcal{K}}^{n})\\ \nonumber &- H(U_{\mathcal{K}}^{n}|W_{\mathcal{K}},Z^{n})-H(Z^{n})\\
\nonumber &\stackrel{(a)}{\geq}
H(W_{\mathcal{K}},U_{\mathcal{K}}^{n})+H(Z^{n}|W_{\mathcal{K}},U_{\mathcal{K}}^{n})-n\epsilon_{n}\\ \nonumber &-H(Z^{n})\\
\nonumber
&\stackrel{(b)}{=}H(W_{\mathcal{K}},U_{\mathcal{K}}^{n})+H(Z^{n}|U_{\mathcal{K}}^{n})-n\epsilon_{n}-H(Z^{n}) \\
\nonumber&\stackrel{(c)}{\geq}
H(U_{\mathcal{K}}^{n})+H(Z^{n}|U_{\mathcal{K}}^{n})- n\epsilon_{n}-H(Z^{n}) \\
\nonumber &= H(U_{\mathcal{K}}^{n})- I(U_{\mathcal{K}}^{n};Z^{n})- n\epsilon_{n}\\
\nonumber &\stackrel{(d)}{\geq} I(U_{\mathcal{K}}^{n};Y^{n})-I(U_{\mathcal{K}}^{n};Z^{n})- n\epsilon_{n}\\
\nonumber &\stackrel{(e)}{\geq} n\sum_{k\in\mathcal{K}}R_{k}-n\epsilon_{n}-n\delta_{1n}-n\delta_{4n}\\ \nonumber &= H(W_{\mathcal{K}})-n\epsilon_{n}-n\delta_{1n}-n\delta_{4n},
\end{IEEEeqnarray}
where $(a)$ follows from Fano's inequality, which states that for
sufficiently large $n$, $H(U_{\mathcal{K}}^{n}|W_{\mathcal{K}},Z^{n})$
$\leq h(P_{we}^{(n)})$ $+nP_{we}^{n}R_{w}\leq n\epsilon_{n}$. Here
$P_{we}^{n}$ denotes the wiretapper's error probability of decoding
$u_{\mathcal{K}}^{n}$ in the case that the bin numbers $w_{\mathcal{K}}$ are known to the eavesdropper and
$R_{w}=I(U_{\mathcal{K}};Z)$.
Since the sum rate is small enough, then $P_{we}^{n}\rightarrow 0$
for sufficiently large $n$. $(b)$ follows from the following Markov
chain: $W_{\mathcal{K}}\rightarrow U_{\mathcal{K}}^{n}\rightarrow$
$Z^{n}$. Hence, we have
$H(Z^{n}|W_{\mathcal{K}},U_{\mathcal{K}}^{n})=H(Z^{n}|U_{\mathcal{K}}^{n})$.
$(c)$ follows from the fact that
$H(W_{\mathcal{K}},U_{\mathcal{K}}^{n})\geq H(U_{\mathcal{K}}^{n})$.
$(d)$ follows from that fact that $H(U_{\mathcal{K}}^{n})\geq
I(U_{\mathcal{K}}^{n};Y^{n})$. $(e)$ follows from the following lemma:

\begin{lem}
Assume $U_{\mathcal{K}}^{n}$, $Y^{n}$ and $Z^{n}$ are generated according to the achievability scheme of Theorem \ref{th1} we then have,
\begin{IEEEeqnarray}{rl}\nonumber \label{eq3}
nI(U_{\mathcal{K}};Y)-n\delta_{1n}\leq I(U_{\mathcal{K}}^{n};Y^{n})\leq nI(U_{\mathcal{K}};Y)+n\delta_{2n}\\ \nonumber \label{eq4}
nI(U_{\mathcal{K}};Z)-n\delta_{3n}\leq I(U_{\mathcal{K}}^{n};Z^{n})\leq nI(U_{\mathcal{K}};Z)+n\delta_{4n},
\end{IEEEeqnarray}
where, $\delta_{1n},\delta_{2n},\delta_{3n},\delta_{4n}\rightarrow 0$ when $n\rightarrow\infty$.
\begin{proof}
Let $A^{(n)}_{\epsilon}(P_{U_{\mathcal{K}},Z})$ denote the set of typical sequences $(u_{\mathcal{K}}^{n},z^{n})$ with respect to $P_{U_{\mathcal{K}},Z}$, and
\begin{IEEEeqnarray}{lr}
\zeta=\left\{
        \begin{array}{ll}
          1, & (u_{\mathcal{K}}^{n},z^{n})\in A^{(n)}_{\epsilon} \\
          0, & \hbox{otherwise}
        \end{array}
      \right.
\end{IEEEeqnarray}
be the corresponding indicator function. We expand $I(U_{\mathcal{K}}^{n};Z^{n},\zeta)$ and $I(U_{\mathcal{K}}^{n},\zeta;Z^{n})$ as follows:
\begin{IEEEeqnarray}{rl}
I(U_{\mathcal{K}}^{n};Z^{n},\zeta)&=I(U_{\mathcal{K}}^{n};Z^{n})+I(U_{\mathcal{K}}^{n};\zeta|Z^{n})\\ \nonumber
&=I(U_{\mathcal{K}}^{n};\zeta)+I(U_{\mathcal{K}}^{n};Z^{n}|\zeta),
\end{IEEEeqnarray}
and
\begin{IEEEeqnarray}{rl}
I(U_{\mathcal{K}}^{n},\zeta;Z^{n})&=I(U_{\mathcal{K}}^{n};Z^{n})+I(\zeta;Z^{n}|U_{\mathcal{K}}^{n})\\ \nonumber
&=I(\zeta;Z^{n})+I(U_{\mathcal{K}}^{n};Z^{n}|\zeta).
\end{IEEEeqnarray}
Therefore, we have
\begin{IEEEeqnarray}{rl}
I(U_{\mathcal{K}}^{n};Z^{n}|\zeta)-I(U_{\mathcal{K}}^{n};\zeta|Z^{n})&\leq I(U_{\mathcal{K}}^{n};Z^{n})\\ \nonumber &\leq I(U_{\mathcal{K}}^{n};Z^{n}|\zeta)+I(\zeta;Z^{n}).
\end{IEEEeqnarray}
Note that $I(\zeta;Z^{n})\leq H(\zeta)\leq 1$ and $I(U_{\mathcal{K}}^{n};\zeta|Z^{n})\leq H(\zeta|Z^{n})\leq H(\zeta)\leq 1$. Thus, the above inequality implies that
\begin{IEEEeqnarray}{lr}\label{eq7}
\sum_{j=0}^{1}P(\zeta=j)I(U_{\mathcal{K}}^{n};Z^{n}|\zeta=j)-1\leq I(U_{\mathcal{K}}^{n};Z^{n})\\ \nonumber\leq \sum_{j=0}^{1}P(\zeta=j)I(U_{\mathcal{K}}^{n};Z^{n}|\zeta=j)+1.
\end{IEEEeqnarray}
According to the joint typicality property, we have
\begin{IEEEeqnarray}{rl}\label{eq8}
0&\leq P(\zeta=1)I(U_{\mathcal{K}}^{n};Z^{n}|\zeta=1)\\ \nonumber&\leq nP((u_{\mathcal{K}}^{n},z^{n})\in A^{(n)}_{\epsilon}(P_{U_{\mathcal{K}},Z} )) \log\|\mathcal{Z}\|\\ \nonumber&\leq n\epsilon_{n}\log \|\mathcal{Z}\|.
\end{IEEEeqnarray}
Now consider the term $P(\zeta=0)I(U_{\mathcal{K}}^{n};Z^{n}|\zeta=0)$. Following the sequence joint typicality properties, we have
\begin{IEEEeqnarray}{rl}\nonumber
(1-\epsilon_{n})I(U_{\mathcal{K}}^{n};Z^{n}|\zeta=0)&\leq P(\zeta=0)I(U_{\mathcal{K}}^{n};Z^{n}|\zeta=0)\\ \label{eq6} &\leq I(U_{\mathcal{K}}^{n};Z^{n}|\zeta=0),
\end{IEEEeqnarray}
where
\begin{IEEEeqnarray}{rl}\nonumber
I(U_{\mathcal{K}}^{n};Z^{n}|\zeta=0)=\sum_{(u_{\mathcal{K}}^{n},z^{n})\in A^{(n)}_{\epsilon}}P(u_{\mathcal{K}}^{n},z^{n})\log\frac{P(u_{\mathcal{K}}^{n},z^{n})}{P(u_{\mathcal{K}}^{n})P(z^{n})}.
\end{IEEEeqnarray}
Since $H(U_{\mathcal{K}},Z)-\epsilon_{n}\leq -\frac{1}{n}\log P(u_{\mathcal{K}}^{n},z^{n})\leq H(U_{\mathcal{K}},Z)+\epsilon_{n}$, then we have,
\begin{IEEEeqnarray}{lr}\nonumber
n\left[-H(U_{\mathcal{K}},Z)+H(U_{\mathcal{K}})+H(Z)-3\epsilon_{n}\right]\leq I(U_{\mathcal{K}}^{n};Z^{n}|\zeta=0)\\ \nonumber\leq n\left[-H(U_{\mathcal{K}},Z)+H(U_{\mathcal{K}})+H(Z)+3\epsilon_{n}\right],
\end{IEEEeqnarray}
or equivalently,
\begin{IEEEeqnarray}{rl}\label{eq5}
n\left[I(U_{\mathcal{K}};Z)-3\epsilon_{n}\right]&\leq I(U_{\mathcal{K}}^{n};Z^{n}|\zeta=0)\\ \nonumber&\leq n\left[I(U_{\mathcal{K}};Z)+3\epsilon_{n}\right].
\end{IEEEeqnarray}
By substituting (\ref{eq5}) into (\ref{eq6}) and then substituting (\ref{eq6}) and (\ref{eq8}) into (\ref{eq7}), we get the desired result,
\begin{IEEEeqnarray}{rl}\nonumber
nI(U_{\mathcal{K}};Z)-n\delta_{1n}\leq I(U_{\mathcal{K}}^{n};Z^{n})\leq nI(U_{\mathcal{K}};Z)+n\delta_{2n},
\end{IEEEeqnarray}
where
\begin{IEEEeqnarray}{lr}
\delta_{1n}=\epsilon_{n}I(U_{\mathcal{K}};Z)+3(1-\epsilon_{n})\epsilon_{n}+\frac{1}{n}\\ \nonumber
\delta_{2n}=3\epsilon_{n}+\epsilon_{n}\log\|\mathcal{Z}\|+\frac{1}{n}.
\end{IEEEeqnarray}
Following the same steps, one can prove the second inequality.
\end{proof}
\end{lem}


\begin{thebibliography}{9}
\bibitem{1}
C. E. Shannon, ``Communication Theory of Secrecy Systems", {\em Bell
System Technical Journal}, vol. 28, pp. 656-715, Oct. 1949.
\bibitem{2}
A. Wyner, ``The Wire-tap Channel", {\em Bell System Technical
Journal}, vol. 54, pp. 1355-1387, 1975
\bibitem{3}
I. Csisz´ar and J. K¨orner, ``Broadcast Channels with Confidential
Messages", {\em IEEE Trans. Inf. Theory}, vol. 24, no. 3, pp.
339-348, May 1978.
\bibitem{4}
E. Tekin and A. Yener, ``The Gaussian multiple access wire-tap channel",
{\em IEEE Trans. Inf. Theory}, vol. 54, no.12, pp.5747-5755, Dec.2008.
\bibitem{5}
E. Tekin and A. Yener, ``The general Gaussian multiple access and two-way
wire-tap channels: Achievable rates and cooperative jamming",
{\em IEEE Trans. Inf. Theory}, vol.54, no.6, pp.2735–2751, Jun. 2008.
\bibitem{6}
X. He, A. Yener, ``Providing Secrecy With Structured Codes: Tools and Applications to Two-User Gaussian Channels",
Submitted to {\em IEEE Trans. Inf. Theory}, available at:
\url{http://arxiv.org/PS_cache/arxiv/pdf/0907/0907.5388v1.pdf}.
\bibitem{7}
M. A. Maddah-Ali, A. S. Motahari, and A. K. Khandani, ``Communication over MIMO X channels: Interference alignment,
decomposition, and performance analysis," {\em IEEE Trans. Inf. Theory,} vol. 54, no. 8, pp. 3457-3470, August 2008 (also see earlier technical reports by the same authors referenced therein).
\bibitem{8}
A. S. Motahari,S. O. Gharan, and A. K. Khandani, ``Real Interference Alignment with Real Numbers", Submitted to
{\em IEEE Trans. Inf. Theory}, available at:
\url{http://arxiv.org/PS_cache/arxiv/pdf/0908/0908.1208v2.pdf}.
\bibitem{9}
W. M. Schmidt, {\em Diophantine approximation}. Berlin, Springer-Verlag, 1980.
\bibitem{10}
G. H. Hardy and E. M. Wright, {\em An introduction to the theory of numbers,} fifth edition, Oxford science publications, 2003.
\end{thebibliography}
\end{document}